\documentclass[aps,pra,twocolumn]{revtex4-1}
\usepackage{amsmath,amssymb,amsthm,graphicx,bbm}
\usepackage[
	bookmarks = false,
	colorlinks = true,
	linkcolor = blue,
	urlcolor= black,
	citecolor = blue]{hyperref}
\newtheorem{theorem}{Theorem}
\newtheorem{lemma}{Lemma}
\newtheorem{corollary}{Corollary}
\newtheorem{definition}{Definition}

\begin{document}
\title{Bound on local minimum-error discrimination of bipartite quantum states}
\author{Donghoon Ha}
\affiliation{Department of Applied Mathematics and Institute of Natural Sciences, Kyung Hee University, Yongin 17104, Republic of Korea}
\author{Jeong San Kim}
\email{freddie1@khu.ac.kr}
\affiliation{Department of Applied Mathematics and Institute of Natural Sciences, Kyung Hee University, Yongin 17104, Republic of Korea}

\begin{abstract}
We consider the optimal discrimination of bipartite quantum states and provide an upper bound for the maximum success probability of optimal local discrimination. We also provide a necessary and sufficient condition for a measurement to realize the upper bound. We further establish a necessary and sufficient condition for this upper bound to be saturated.  Finally, we illustrate our results using an example.
\end{abstract}
\maketitle
\indent Quantum state discrimination is one of the fundamental tasks in quantum information processing \cite{chef2000,barn20091,berg2010}.
In discriminating orthogonal quantum states, there is always a measurement of perfect discrimination. On the other hand, 
non-orthogonal quantum states cannot be perfectly discriminated by means of any measurement. 
For this reason, there has been a huge amount of research effort focused on finding good state-discriminating strategies \cite{bae2015}.\\
\indent In discriminating multiparty quantum states,
it is known that some optimal state discrimination cannot be realized only by \emph{local operations and classical communication} (LOCC) \cite{pere1991,benn19991,duan2007,chit2013}.
To characterize the limitation of LOCC discrimination,
many studies have been contributed to optimal local discrimination of multiparty quantum states \cite{ghos2001,walg2002,fan2004,duan2009,
chit20141,band2015,band2021}.
Nevertheless, due to the difficulty of mathematical characterization for LOCC, 
it is still a hard task to realize optimal local discrimination.\\
\indent One efficient way to handle this difficulty is to investigate possible upper bounds for the maximum success probability of optimal local discrimination.
Moreover, establishing good conditions on measurements realizing such upper bounds is also important for a better understanding of optimal local discrimination.\\
\indent Here, we consider bipartite quantum state discrimination and provide an upper bound for the maximum success probability of optimal local discrimination. We also provide a necessary and sufficient condition for a measurement to realize the upper bound. Moreover, we establish a necessary and sufficient condition for this upper bound to be saturated; it is equal to the maximum success probability of optimal local discrimination. Finally, we illustrate our results using an example.\\
\indent In bipartite quantum systems, a state is represented by a density operator $\rho$, that is, a Hermitian operator having positive semidefiniteness $\rho\succeq0$ and unit trace $\mathrm{Tr}\rho=1$, acting on a bipartite complex Hilbert space $\mathcal{H}=\mathbb{C}^{d_{1}}\otimes\mathbb{C}^{d_{2}}$.
A measurement is described by a \emph{positive operator-valued measure} (POVM) $\{M_{i}\}_{i}$ that is a set of Hermitian operators $M_{i}$ on $\mathcal{H}$ satisfying 
positive semidefiniteness $M_{i}\succeq0$ for all $i$ and the
completeness relation $\sum_{i}M_{i}=\mathbbm{1}$, where $\mathbbm{1}$ is the identity operator on $\mathcal{H}$. 
When a measurement $\{M_{i}\}_{i}$ is performed for input state $\rho$, the probability of obtaining measurement outcome with respect to $M_{j}$ is $\mathrm{Tr}(\rho M_{j})$.

\begin{definition}
A Hermitian operator $E$ on $\mathcal{H}$ is called positive partial transpose(PPT) if its partial transposition, denoted $E^{\rm PT}$,
is positive semidefinite\cite{pere1996,horo1996,pptp}. Similarly, we say that a set of Hermitian operators $\{E_{i}\}_{i}$ is PPT if $E_{i}$ is PPT for all $i$. 
\end{definition}

\indent A \emph{LOCC measurement} is a measurement that can be realized by LOCC. We note that every LOCC measurement is a PPT measurement \cite{chit20142}.\\
\indent Throughout this paper, we only consider
the situation of discriminating $n$ \emph{bipartite} quantum states $\rho_{1},\ldots,\rho_{n}$ in which the state $\rho_{i}$ is prepared with the probability $\eta_{i}$.
We denote this situation as an ensemble $\mathcal{E}=\{\eta_{i},\rho_{i}\}_{i=1}^{n}$.\\
\indent Let us consider the quantum state discrimination of $\mathcal{E}$ using a measurement $\{M_{i}\}_{i=1}^{n}$ where the click of $M_{i}$ means the detection of $\rho_{i}$.
The \emph{minimum-error discrimination} \cite{hels1976,hole1979,yuen1975} of $\mathcal{E}$ is to achieve the minimum error in correctly guessing the prepared state.
Equivalently, the minimum-error discrimination of $\mathcal{E}$ is to achieve the so-called
\emph{guessing probability} of $\mathcal{E}=\{\eta_{i},\rho_{i}\}_{i=1}^{n}$, defined as
\begin{equation}\label{eq:pgdef}
p_{\rm G}(\mathcal{E})=\max_{\rm POVM}\sum_{i=1}^{n}\eta_{i}\mathrm{Tr}(\rho_{i}M_{i}),
\end{equation}
where the maximum is taken over all possible POVMs.
The POVMs providing the optimal success probability $p_{\rm G}(\mathcal{E})$ in Eq.~\eqref{eq:pgdef} can be verified from the following conditions\cite{hole1979,yuen1975,barn20092,bae2013}:
\begin{subequations}
\begin{eqnarray}
\sum_{j=1}^{n}\eta_{j}\rho_{j}M_{j}-\eta_{i}\rho_{i}\succeq0~\forall i=1,\ldots,n,\label{eq:nscfme}\\[2mm]
M_{i}(\eta_{i}\rho_{i}-\eta_{j}\rho_{j})M_{j}=0~\forall i,j=1,\ldots,n.\label{eq:oncfme}
\end{eqnarray}
\end{subequations}
Note that Condition \eqref{eq:nscfme} is a necessary and sufficient condition for a measurement $\{M_{i}\}_{i=1}^{n}$ to realize $p_{\rm G}(\mathcal{E})$, whereas 
Condition \eqref{eq:oncfme} is a necessary but not sufficient condition for a POVM $\{M_{i}\}_{i=1}^{n}$ to provide $p_{\rm G}(\mathcal{E})$.\\
\indent When the available measurements are limited to PPT POVMs, we denote the maximum success probability by
\begin{equation}\label{eq:pptdef}
p_{\rm PPT}(\mathcal{E})=\max_{\substack{\rm PPT\\ \rm POVM}}\sum_{i=1}^{n}\eta_{i}\mathrm{Tr}(\rho_{i}M_{i}).
\end{equation}
We denote by $p_{\rm L}(\mathcal{E})$ the maximum of success probability
that can be obtained by using LOCC measurements; that is,
\begin{equation}\label{eq:pldef}
p_{\rm L}(\mathcal{E})=\max_{\rm LOCC}\sum_{i=1}^{n}\eta_{i}\mathrm{Tr}(\rho_{i}M_{i}).
\end{equation}
From the definitions of $p_{\rm G}(\mathcal{E})$ and $p_{\rm PPT}(\mathcal{E})$,
$p_{\rm G}(\mathcal{E})$ is obviously an upper bound of $p_{\rm PPT}(\mathcal{E})$.
Moreover, $p_{\rm L}(\mathcal{E})$ is a lower bound of $p_{\rm PPT}(\mathcal{E})$
because all LOCC measurements are PPT \cite{chit20142}. Thus, we have
\begin{equation}
p_{\rm L}(\mathcal{E})\leqslant p_{\rm PPT}(\mathcal{E})\leqslant p_{\rm G}(\mathcal{E}).
\end{equation}
We also note that $p_{\rm L}(\mathcal{E})=p_{\rm PPT}(\mathcal{E})$ if and only if 
there exists a LOCC measurement realizing $p_{\rm PPT}(\mathcal{E})$ since both $p_{\rm PPT}(\mathcal{E})$ and $p_{\rm L}(\mathcal{E})$ have the same objective function for maximization.\\
\indent For a given ensemble $\mathcal{E}=\{\eta_{i},\rho_{i}\}_{i=1}^{n}$, let us consider the maximum quantity
\begin{equation}\label{eq:qgdef}
q_{\rm G}(\mathcal{E})=\max_{\rm POVM}\sum_{i=1}^{n}\eta_{i}\mathrm{Tr}(\rho_{i}^{\rm PT}M_{i})
\end{equation}
over all possible POVMs.
The following lemma shows that $q_{\rm G}(\mathcal{E})$ in Eq.~\eqref{eq:qgdef} is an upper bound of $p_{\rm PPT}(\mathcal{E})$:

\begin{lemma}\label{lem:pptqg}
For a bipartite quantum state ensemble $\mathcal{E}=\{\eta_{i},\rho_{i}\}_{i=1}^{n}$, 
\begin{equation}\label{eq:ubppt}
p_{\rm PPT}(\mathcal{E})\leqslant q_{\rm G}(\mathcal{E}),
\end{equation}
where the equality holds if and only if there exists a PPT measurement realizing $q_{\rm G}(\mathcal{E})$.
\end{lemma}
\begin{proof}
Inequality~\eqref{eq:ubppt} holds because
\begin{eqnarray}
q_{\rm G}(\mathcal{E})&=&\max_{\substack{M_{i}\succeq0\,\forall i\\ \sum_{i=1}^{n}M_{i}=\mathbbm{1}}}\sum_{i=1}^{n}\eta_{i}\mathrm{Tr}(\rho_{i}M_{i}^{\rm PT})\nonumber\\
&=&\max_{\substack{M_{i}^{\rm PT}\succeq0\,\forall i\\ \sum_{i=1}^{n}M_{i}=\mathbbm{1}}}\sum_{i=1}^{n}\eta_{i}\mathrm{Tr}(\rho_{i}M_{i})
\geqslant p_{\rm PPT}(\mathcal{E}),\label{eq:qgubr}
\end{eqnarray}
where the first equality follows from $\mathrm{Tr}(AB)=\mathrm{Tr}(A^{\rm PT}B^{\rm PT})$ for any two operators $A$ and $B$, the second equality holds due to $\mathbbm{1}=\mathbbm{1}^{\rm PT}$, and the inequality is from the fact that a PPT POVM 
$\{M_{i}\}_{i=1}^{n}$ implies $M_{i}\succeq0$ for all $i$
along with $M_{i}^{\rm PT}\succeq0$ for all $i$ 
and $\sum_{i=1}^{n}M_{i}=\mathbbm{1}$. \\
\indent If $p_{\rm PPT}(\mathcal{E})=q_{\rm G}(\mathcal{E})$, then
\begin{equation}
\begin{array}{rcl}
q_{\rm G}(\mathcal{E})=p_{\rm PPT}(\mathcal{E})
&=&\sum_{i=1}^{n}\eta_{i}\mathrm{Tr}(\rho_{i}M_{i}) \\[2mm]
&=&\sum_{i=1}^{n}\eta_{i}\mathrm{Tr}(\rho_{i}^{\rm PT}M_{i}^{\rm PT})
\end{array}
\end{equation}
for some PPT POVM $\{M_{i}\}_{i=1}^{n}$. Since $\{M_{i}^{\rm PT}\}_{i=1}^{n}$ is also a PPT POVM, there exists a PPT measurement giving $q_{\rm G}(\mathcal{E})$. Conversely, if $\{M_{i}\}_{i=1}^{n}$ is a PPT POVM providing $q_{\rm G}(\mathcal{E})$, then
\begin{equation}\label{eq:pgbr}
\begin{array}{rcl}
p_{\rm PPT}(\mathcal{E})&\leqslant& q_{\rm G}(\mathcal{E})
=\sum_{i=1}^{n}\eta_{i}\mathrm{Tr}(\rho_{i}^{\rm PT}M_{i})\\[2mm]
&=&\sum_{i=1}^{n}\eta_{i}\mathrm{Tr}(\rho_{i}M_{i}^{\rm PT})
\leqslant p_{\rm PPT}(\mathcal{E})
\end{array}
\end{equation}
where the first inequality is from Inequality \eqref{eq:qgubr} and the second inequality follows from the fact that $\{M_{i}^{\rm PT}\}_{i=1}^{n}$ is also a PPT POVM. Thus, $p_{\rm PPT}(\mathcal{E})=q_{\rm G}(\mathcal{E})$.
\end{proof}

\begin{corollary}\label{cor:plqg}
For a bipartite quantum state ensemble $\mathcal{E}=\{\eta_{i},\rho_{i}\}_{i=1}^{n}$, 
\begin{equation}\label{eq:pleqqg}
p_{\rm L}(\mathcal{E})=q_{\rm G}(\mathcal{E})
\end{equation}
if and only if there is a POVM $\{M_{i}\}_{i=1}^{n}$ giving $q_{\rm G}(\mathcal{E})$ such that
$\{M_{i}^{\rm PT}\}_{i=1}^{n}$ is a LOCC measurement.
\end{corollary}
\begin{proof}
\indent Suppose that 
$\{M_{i}\}_{i=1}^{n}$ is a POVM giving $q_{\rm G}(\mathcal{E})$ and $\{M_{i}^{\rm PT}\}_{i=1}^{n}$ is a LOCC measurement. 
Since $\{M_{i}\}_{i=1}^{n}$ is a PPT POVM, $p_{\rm PPT}(\mathcal{E})=q_{\rm G}(\mathcal{E})$ due to Lemma~\ref{lem:pptqg}.
Also, $\{M_{i}^{\rm PT}\}_{i=1}^{n}$ gives $p_{\rm PPT}(\mathcal{E})$ in Eq.~\eqref{eq:pptdef} because
\begin{equation}
\begin{array}{rcl}
\sum_{i=1}^{n}\eta_{i}\mathrm{Tr}(\rho_{i}M_{i}^{\rm PT})
&=&\sum_{i=1}^{n}\eta_{i}\mathrm{Tr}(\rho_{i}^{\rm PT}M_{i})\\[2mm]
&=&q_{\rm G}(\mathcal{E})=p_{\rm PPT}(\mathcal{E}).
\end{array}
\end{equation}
The existence of a LOCC measurement giving $p_{\rm PPT}(\mathcal{E})$ implies $p_{\rm L}(\mathcal{E})=p_{\rm PPT}(\mathcal{E})$. Thus, Eq.~\eqref{eq:pleqqg} holds.\\
\indent Conversely, if Eq.~\eqref{eq:pleqqg} is satisfied, then
\begin{equation}
\begin{array}{rcl}
q_{\rm G}(\mathcal{E})=p_{\rm L}(\mathcal{E})&=&\sum_{i=1}^{n}\eta_{i}\mathrm{Tr}(\rho_{i}\tilde{M}_{i})\\[2mm]
&=&\sum_{i=1}^{n}\eta_{i}\mathrm{Tr}(\rho_{i}^{\rm PT}\tilde{M}_{i}^{\rm PT}),
\end{array}
\end{equation}
where $\{\tilde{M}_{i}\}_{i=1}^{n}$ is a LOCC measurement realizing $p_{\rm L}(\mathcal{E})$ in Eq.~\eqref{eq:pldef}. 
Since every LOCC measurement is PPT,
it follows that $\{\tilde{M}_{i}^{\rm PT}\}_{i=1}^{n}$ is a POVM.
Thus, $\{M_{i}\}_{i=1}^{n}$ with $M_{i}=\tilde{M}_{i}^{\rm PT}$ for all $i$ is a POVM giving $q_{\rm G}(\mathcal{E})$ such that
$\{M_{i}^{\rm PT}\}_{i=1}^{n}$ is a LOCC measurement.
\end{proof}

\indent For a given state ensemble $\mathcal{E}=\{\eta_{i},\rho_{i}\}_{i=1}^{n}$, the following theorem provides a necessary and sufficient condition on a measurement $\{M_{i}\}_{i=1}^{n}$ to realize $q_{\rm G}(\mathcal{E})$ in Eq.~\eqref{eq:qgdef}.
\begin{theorem}\label{thm:mnsc}
For a bipartite quantum state ensemble $\mathcal{E}=\{\eta_{i},\rho_{i}\}_{i=1}^{n}$, a POVM $\{M_{i}\}_{i=1}^{n}$ gives $q_{\rm G}(\mathcal{E})$ if and only if it satisfies
\begin{equation}\label{eq:nscqg}
\sum_{j=1}^{n}\eta_{j}\rho_{j}^{\rm PT}M_{j}-\eta_{i}\rho_{i}^{\rm PT}\succeq0~\forall i=1,\ldots,n.
\end{equation}
Moreover, if a POVM $\{M_{i}\}_{i=1}^{n}$ realizes $q_{\rm G}(\mathcal{E})$, then
\begin{equation}\label{eq:oncqg}
M_{i}(\eta_{i}\rho_{i}^{\rm PT}-\eta_{j}\rho_{j}^{\rm PT})M_{j}=0~\forall i,j=1,\ldots,n.
\end{equation}
\end{theorem}
\begin{proof}
\indent To prove the sufficiency of the first statement, we suppose that $\{M_{i}\}_{i=1}^{n}$ is a POVM satisfying Condition~\eqref{eq:nscqg}. For any POVM $\{M_{i}'\}_{i=1}^{n}$, we have
\begin{eqnarray}
\begin{array}{rcl}
&&\sum_{j=1}^{n}\eta_{j}\mathrm{Tr}(\rho_{j}^{\rm PT}M_{j})-\sum_{k=1}^{n}\eta_{k}\mathrm{Tr}(\rho_{k}^{\rm PT}M_{k}')\\[3mm]
&=&
\mathrm{Tr}\big[\sum_{j=1}^{n}\eta_{j}\rho_{j}^{\rm PT}M_{j}\big(\sum_{i=1}^{n}M_{i}'\big)\big]\\[2mm]
&&-\sum_{k=1}^{n}\mathrm{Tr}(\eta_{k}\rho_{k}^{\rm PT}M_{k}')\\[3mm]
&=&\sum_{i=1}^{n}\mathrm{Tr}\big[
\big(\sum_{j=1}^{n}\eta_{j}\rho_{j}^{\rm PT}M_{j}-\eta_{i}\rho_{i}^{\rm PT}\big)M_{i}'
\big]\geqslant0,
\end{array}
\end{eqnarray}
where the first equality follows from $\sum_{i=1}^{n}M_{i}'=\mathbbm{1}$ and the inequality is from $M_{i}'\succeq0$ for all $i$ and Condition~\eqref{eq:nscqg}.
Thus, the definition of $q_{\rm G}(\mathcal{E})$ leads us to 
\begin{equation}
\sum_{i=1}^{n}\eta_{i}\mathrm{Tr}(\rho_{i}^{\rm PT}M_{i})=q_{\rm G}(\mathcal{E}), 
\end{equation}
which proves the sufficiency of the first statement.\\
\indent To prove the second statement along with the necessity of the first statement, we assume that $\{M_{i}\}_{i=1}^{n}$ is a POVM providing $q_{\rm G}(\mathcal{E})$. 
We first show the positive semidefiniteness of the following Hermitian operators
\begin{equation}\label{eq:hidef}
\begin{array}{rcl}
H_{i}&=&\frac{1}{2}\sum_{j=1}^{n}(\eta_{j}\rho_{j}^{\rm PT}M_{j}+\eta_{j}M_{j}\rho_{j}^{\rm PT})\\[2mm]
&&-\eta_{i}\rho_{i}^{\rm PT},~i=1,\ldots,n.
\end{array}
\end{equation}
To show it, we first suppose $\langle v|H_{1}|v\rangle<0$ for some unit vector $|v\rangle$ and lead to a contradiction.
For $0<\epsilon<1$, let us consider the following POVM $\{M_{i}^{(\epsilon)}\}_{i=1}^{n}$,
\begin{equation}
\begin{array}{rcl}
M_{i}^{(\epsilon)}&=&(\mathbbm{1}-\epsilon|v\rangle\!\langle v|)M_{i}(\mathbbm{1}-\epsilon|v\rangle\!\langle v|)\\[2mm]
&&+\epsilon(2-\epsilon)|v\rangle\!\langle v|\delta_{i1},~i=1,\ldots,n,
\end{array}
\end{equation}
where $\delta_{ij}$ is the Kronecker delta.\\
\indent From a straightforward calculation, we can easily see that
\begin{eqnarray}\label{eq:scws}
&&\sum_{j=1}^{n}\eta_{j}\mathrm{Tr}(\rho_{j}^{\rm PT}M_{j}^{(\epsilon)})
-\sum_{i=1}^{n}\eta_{i}\mathrm{Tr}(\rho_{i}^{\rm PT}M_{i})
=-2\epsilon\langle v|H_{1}|v\rangle\nonumber\\
&&+\epsilon^{2}\big[\sum_{i=1}^{n}\eta_{i}\langle v|\rho_{i}^{\rm PT}|v\rangle\!\langle v|M_{i}|v\rangle
-\eta_{1}\langle v|\rho_{1}^{\rm PT}|v\rangle\big].
\end{eqnarray}
For given real numbers $a$ and $b$ with $a>0$, we note that there is $\epsilon\in(0,1)$ such that $a\epsilon+b\epsilon^{2}>0$. Thus, the right-hand side of Eq.~\eqref{eq:scws} is positive for some $\epsilon\in(0,1)$. This contradicts the assumption that $\{M_{i}\}_{i=1}^{n}$ realizes $q_{\rm G}(\mathcal{E})$, therefore $H_{1}\succeq0$. Since the choice of $H_{1}$ with a negative eigenvalue can be arbitrary, $H_{i}\succeq0$ for all $i$.\\
\indent Now, we show the satisfaction of Conditions~\eqref{eq:nscqg} and \eqref{eq:oncqg}. It is straightforward to verify that
\begin{equation}\label{eq:shimi}
2\sum_{i=1}^{n}H_{i}M_{i}=\sum_{i=1}^{n}\eta_{i}M_{i}\rho_{i}^{\rm PT}
-\sum_{j=1}^{n}\eta_{j}\rho_{j}^{\rm PT}M_{j}.
\end{equation}
Since the right-hand side of Eq.~\eqref{eq:shimi} is traceless, we have
\begin{equation}\label{eq:trless}
\sum_{i=1}^{n}\mathrm{Tr}(H_{i}M_{i})=0,
\end{equation}
which implies
\begin{equation}\label{eq:hmmh}
H_{i}M_{i}=M_{i}H_{i}=0~\forall i=1,\ldots,n
\end{equation}
due to the positive semidefiniteness of $H_{i}$ and $M_{i}$ for all $i$.
Equations~\eqref{eq:shimi} and \eqref{eq:hmmh} lead us to
\begin{equation}\label{eq:emrerm}
\sum_{i=1}^{n}\eta_{i}M_{i}\rho_{i}^{\rm PT}
=\sum_{i=1}^{n}\eta_{i}\rho_{i}^{\rm PT}M_{i}.
\end{equation}
Applying Eq.~\eqref{eq:emrerm} to Eq.~\eqref{eq:hidef}, we can show that
\begin{equation}
\sum_{j=1}^{n}\eta_{j}\rho_{j}^{\rm PT}M_{j}-\eta_{i}\rho_{i}^{\rm PT}=H_{i}\succeq0~\forall i=1,\ldots,n,
\end{equation}
therefore Condition~\eqref{eq:nscqg} holds.
That is, the necessity of the first statement is true.\\
\indent Moreover, Eq.~\eqref{eq:hmmh} leads us to
\begin{equation}
\begin{array}{rcl}
&&M_{i}(\eta_{i}\rho_{i}^{\rm PT}-\eta_{j}\rho_{j}^{\rm PT})M_{j}\\[2mm]
&=&M_{i}H_{j}M_{j}-M_{i}H_{i}M_{j}=0~\forall i,j=1,\ldots,n,
\end{array}
\end{equation}
this is, Condition~\eqref{eq:oncqg} is satisfied.
Therefore, the second statement is also true.
\end{proof}

\indent From Corollary~\ref{cor:plqg} and Theorem~\ref{thm:mnsc}, we have the following corollary.

\begin{corollary}\label{cor:plqgnc}
For a bipartite quantum state ensemble $\mathcal{E}=\{\eta_{i},\rho_{i}\}_{i=1}^{n}$, 
\begin{equation}
p_{\rm L}(\mathcal{E})=q_{\rm G}(\mathcal{E})
\end{equation}
if and only if there is a POVM $\{M_{i}\}_{i=1}^{n}$ satisfying Condition~\eqref{eq:nscqg} such that
$\{M_{i}^{\rm PT}\}_{i=1}^{n}$ is a LOCC measurement.
\end{corollary}

\indent For any integer $d\geqslant 2$, let us consider the two-qu$d$it state ensemble $\mathcal{E}=\{\eta_{i,j}^{(k)},\rho_{i,j}^{(k)}\}_{i,j,k}$ consisting of $2d(d-1)$ states with equal prior probability, 
\begin{eqnarray}\label{eq:exerho}
\begin{array}{ll}
\eta_{i,j}^{(k)}=\frac{1}{2d(d-1)},~\rho_{i,j}^{(k)}=\lambda|\Psi_{i,j}^{(k)}\rangle\!\langle\Psi_{i,j}^{(k)}|+(1-\lambda)\sigma,\\[2mm]
i,j\in\{0,1,\ldots,d-1\}~\mbox{with}~i<j,~k=1,2,3,4,
\end{array}
\end{eqnarray}
where $0<\lambda\leqslant1$,
$\sigma$ is an arbitrary two-qudit state, and 
\begin{equation}
\begin{array}{ll}
|\Psi_{i,j}^{(1)}\rangle=\frac{1}{\sqrt{2}}(|i\rangle\otimes|i\rangle+|j\rangle\otimes|j\rangle),\\
|\Psi_{i,j}^{(2)}\rangle=\frac{1}{\sqrt{2}}(|i\rangle\otimes|i\rangle-|j\rangle\otimes|j\rangle),\\
|\Psi_{i,j}^{(3)}\rangle=\frac{1}{\sqrt{2}}(|i\rangle\otimes|j\rangle+|j\rangle\otimes|i\rangle),\\
|\Psi_{i,j}^{(4)}\rangle=\frac{1}{\sqrt{2}}(|i\rangle\otimes|j\rangle-|j\rangle\otimes|i\rangle).
\end{array}
\end{equation}
For a POVM $\{M_{i,j}^{(k)}\}_{i,j,k}$ with
\begin{equation}
\begin{array}{l}
M_{i,j}^{(1)}=\frac{1}{d-1}|\Psi_{i,j}^{(1)}\rangle\!\langle\Psi_{i,j}^{(1)}|,~
M_{i,j}^{(3)}=|\Psi_{i,j}^{(3)}\rangle\!\langle\Psi_{i,j}^{(3)}|,\\[2mm]
M_{i,j}^{(2)}=\frac{1}{d-1}|\Psi_{i,j}^{(2)}\rangle\!\langle\Psi_{i,j}^{(2)}|,~
M_{i,j}^{(4)}=|\Psi_{i,j}^{(4)}\rangle\!\langle\Psi_{i,j}^{(4)}|,\\[2mm]
\end{array}
\end{equation}
Condition \eqref{eq:nscfme} holds, that is,
\begin{equation}
\begin{array}{rcl}
&&\sum_{i',j',k'}\eta_{i',j'}^{(k')}\rho_{i',j'}^{(k')}M_{i',j'}^{(k')}-\eta_{i,j}^{(k)}\rho_{i,j}^{(k)}\\[1mm]
&=&\frac{\lambda}{2d(d-1)}\big(\mathbbm{1}-|\Psi_{i,j}^{(k)}\rangle\!\langle\Psi_{i,j}^{(k)}|\big)\succeq0,~\forall i,j,k.
\end{array}
\end{equation}
Therefore, the optimal success probability $p_{\rm G}(\mathcal{E})$ in Eq.~\eqref{eq:pgdef} is
\begin{equation}\label{eq:expg}
p_{\rm G}(\mathcal{E})=
\sum_{i,j,k}\eta_{i,j}^{(k)}\mathrm{Tr}(\rho_{i,j}^{(k)}M_{i,j}^{(k)})
=\frac{1+\lambda (d^{2}-1)}{2d(d-1)}.
\end{equation}
\indent To obtain $q_{\rm G}(\mathcal{E})$ in Eq.~\eqref{eq:qgdef}, we use a POVM $\{M_{i,j}^{(k)}\}_{i,j,k}$ that consists of
\begin{equation}\label{eq:exm}
\begin{array}{ll}
M_{i,j}^{(1)}=\frac{1}{d-1}|i\rangle\!\langle i|\otimes|i\rangle\!\langle i|,&
M_{i,j}^{(3)}=|i\rangle\!\langle i|\otimes|j\rangle\!\langle j|,\\[2mm]
M_{i,j}^{(2)}=\frac{1}{d-1}|j\rangle\!\langle j|\otimes|j\rangle\!\langle j|,&
M_{i,j}^{(4)}=|j\rangle\!\langle j|\otimes|i\rangle\!\langle i|.
\end{array}
\end{equation}
This POVM satisfies Condition~\eqref{eq:nscqg} because
\begin{equation}
\begin{array}{rcl}
&&\sum_{i',j',k'}\eta_{i',j'}^{(k')}\rho_{i',j'}^{(k')\,\rm PT}M_{i',j'}^{(k')}-\eta_{i,j}^{(k)}\rho_{i,j}^{(k)\,\rm PT}\\[2mm]
&=&\frac{\lambda}{4d(d-1)}\big(\mathbbm{1}-\mathbbm{1}_{i,j}+2|\Psi_{i,j}^{(5-k)}\rangle\!\langle\Psi_{i,j}^{(5-k)}|\big)\\[2mm]
&\succeq&0~\forall i,j,k,
\end{array}
\end{equation}
where 
\begin{equation}
\mathbbm{1}_{i,j}=(|i\rangle\!\langle i|+|j\rangle\!\langle j|)\otimes(|i\rangle\!\langle i|+|j\rangle\!\langle j|)
\end{equation}
and the equality is due to 
\begin{eqnarray}
\begin{array}{l}
|\Psi_{i,j}^{(k)}\rangle\!\langle\Psi_{i,j}^{(k)}|^{\rm PT}=\frac{1}{2}\mathbbm{1}_{i,j}
-|\Psi_{i,j}^{(5-k)}\rangle\!\langle\Psi_{i,j}^{(5-k)}|,\\[2mm]
|\Psi_{i,j}^{(5-k)}\rangle\!\langle\Psi_{i,j}^{(5-k)}|M_{i,j}^{(k)}=0
~\forall i,j,k.
\end{array}
\end{eqnarray}
Thus, Theorems~\ref{thm:mnsc} leads us to
\begin{equation}\label{eq:exqg}
q_{\rm G}(\mathcal{E})=
\sum_{i,j,k}\eta_{i,j}^{(k)}\mathrm{Tr}(\rho_{i,j}^{(k)}M_{i,j}^{(k)})
=\frac{2+\lambda (d^{2}-2)}{4d(d-1)}.
\end{equation}
\indent Moreover, the POVM $\{M_{i,j}^{(k)}\}_{i,j,k}$ in Eq.~\eqref{eq:exm} is a LOCC measurement since it can be implemented by performing the same local measurement $\{|l\rangle\!\langle l|\}_{l=0}^{d-1}$ on two subsystems.
Thus, Corollary~\ref{cor:plqgnc} and Eq.~\eqref{eq:exqg} lead us to
\begin{equation}\label{eq:expqg}
\begin{array}{c}
p_{\rm L}(\mathcal{E})=q_{\rm G}(\mathcal{E})=\frac{2+\lambda (d^{2}-2)}{4d(d-1)}
=p_{\rm G}(\mathcal{E})-\frac{\lambda d}{4(d-1)}.
\end{array}
\end{equation}
In the case of $d=2$, Eqs.~\eqref{eq:expg} and \eqref{eq:expqg}
coincide with the existing results in Ref.~\cite{band2021}.\\
\indent In this paper, we have considered the situation of discriminating bipartite quantum states, and provided an upper bound $q_{\rm G}(\mathcal{E})$ for the maximum success probability of optimal local discrimination $p_{\rm L}(\mathcal{E})$(Lemma~\ref{lem:pptqg}).
We have further established a necessary and sufficient condition for a measurement to realize $q_{\rm G}(\mathcal{E})$(Theorem~\ref{thm:mnsc}).
Moreover, we have provided the equality condition between $q_{\rm G}(\mathcal{E})$ and $p_{\rm L}(\mathcal{E})$  (Corollaries~\ref{cor:plqg} and \ref{cor:plqgnc}).
Finally, we have illustrated the effectiveness of our results through an example.\\
\indent We note that finding $p_{\rm G}(\mathcal{E})$ or $q_{\rm G}(\mathcal{E})$ in discriminating separable quantum states can be useful in studying the nonlocal phenomenon of separable quantum states, namely \emph{nonlocality without entanglement}(NLWE) \cite{pere1991,benn19991}.
For the minimum-error discrimination of a separable state ensemble $\{\eta_{i},\rho_{i}\}_{i=1}^{n}$,
NLWE occurs if the guessing probability $p_{\rm G}(\mathcal{E})$ cannot be achieved only by LOCC,
that is, $p_{\rm L}(\mathcal{E})<p_{\rm G}(\mathcal{E})$.
From Lemma~\ref{lem:pptqg}, $q_{\rm G}(\mathcal{E})<p_{\rm G}(\mathcal{E})$ implies $p_{\rm L}(\mathcal{E})<p_{\rm G}(\mathcal{E})$, therefore the occurrence of NLWE.
Moreover, even if $q_{\rm G}(\mathcal{E})>p_{\rm G}(\mathcal{E})$, we can show the NLWE phenomenon in terms of $\{\eta_{i},\rho_{i}^{\rm PT}\}_{i=1}^{n}$ because 
the partial transposition of any separable state is another separable state and the roles of $p_{\rm G}(\mathcal{E})$ and $q_{\rm G}(\mathcal{E})$ are interchanged for the minimum-error discrimination of $\{\eta_{i},\rho_{i}^{\rm PT}\}_{i=1}^{n}$.\\
\indent It is an interesting future work to investigate 
good conditions of optimal local discrimination in multiparty quantum systems having more than two parties.
It is also natural to ask if our results are still valid  
for other optimal discrimination strategies other than minimum-error discrimination.

This work was supported by Basic Science Research Program(NRF-2020R1F1A1A010501270) and Quantum Computing Technology Development Program(NRF-2020M3E4A1080088) through the National Research Foundation of Korea(NRF) grant funded by the Korea government(Ministry of Science and ICT).



\begin{thebibliography}{61}
\bibitem{pptp}  
PPT property does not depend on the choice of basis or the subsystem to be transposed. 
For simplicity, we consider the standard basis and the second subsystem throughout this paper.
\end{thebibliography}
\end{document}